\documentclass[11pt]{article}

\usepackage[T5]{fontenc}

\usepackage{amsmath,amssymb,amsthm,algpseudocode,float,mathtools,thmtools,thm-restate,scalerel,stackengine,verbatim}
\usepackage{algorithm,algpseudocode}
  
\usepackage[labelfont=bf]{caption}
 
\stackMath
\newcommand\reallywidehat[1]{%
\savestack{\tmpbox}{\stretchto{%
  \scaleto{%
      \scalerel*[\widthof{\ensuremath{#1}}]{\kern-.6pt\bigwedge\kern-.6pt}%
          {\rule[-\textheight/2]{1ex}{\textheight}}
            }{\textheight}%
            }{0.5ex}}%
            \stackon[1pt]{#1}{\tmpbox}%
            }
\usepackage[margin=1in]{geometry}
\usepackage{multirow}
\usepackage{subcaption}
\usepackage{afterpage}
\usepackage{bm}
\usepackage[hidelinks]{hyperref}
\usepackage{tikz}
\usetikzlibrary{decorations.pathreplacing, positioning}

\usepackage{bbm}

\newcommand{\wt}[1]{\widetilde{#1}}

\newcommand{\bef}{\preccurlyeq}

\DeclarePairedDelimiter\bra{\langle}{|}
\DeclarePairedDelimiter\ket{|}{\rangle}
\DeclarePairedDelimiterX\braket[2]{\langle}{\rangle}{#1\delimsize\vert#2}

\DeclarePairedDelimiter\ceil{\lceil}{\rceil}
\DeclarePairedDelimiter\norm{\lVert}{\rVert}
\DeclarePairedDelimiter\abs{|}{|}
\DeclarePairedDelimiter\brac{\lbrack}{\rbrack}
\DeclarePairedDelimiter\set{\lbrace}{\rbrace}
\DeclarePairedDelimiter\paren{\lparen}{\rparen}

\newcommand{\E}[2][]{\operatorname*{\mathbb{E}}_{#1 }\brac*{#2}}

\newcommand{\Var}[2][]{\operatorname*{\normalfont{\text{Var}}}_{#1 }\paren*{#2}}
\newcommand{\bO}[1]{\operatorname*{O}\paren*{#1}}
\newcommand{\bOt}[1]{\operatorname*{\wt{O}}\paren*{#1}}
\newcommand{\lO}[1]{\operatorname*{o}\paren*{#1}}
\newcommand{\bOm}[1]{\operatorname*{\Omega}\paren*{#1}}
\newcommand{\lOm}[1]{\operatorname*{\omega}\paren*{#1}}
\newcommand{\bT}[1]{\operatorname*{\Theta}\paren*{#1}}

\newcommand{\bb}{\mathbf{b}}
\newcommand{\Xb}{\mathbf{X}}

\newcommand{\Oc}{\mathcal{O}}

\newcommand{\degb}[1]{d_{#1}^{\rightarrow}}
\newcommand{\dedge}[1]{\overrightarrow{#1}}
\newcommand{\tgk}{t^{>k}}
\newcommand{\tlk}{t^{<k}}
\newcommand{\Tgk}{T^{>k}}
\newcommand{\Tlk}{T^{<k}}

\newcommand{\Ibb}{\mathbb{I}}

\newcommand{\return}{\textbf{return}~}

\newcommand{\Cbb}{\mathbb{R}}
\newcommand{\Rbb}{\mathbb{R}}
\newcommand{\Zbb}{\mathbb{Z}}

\newcommand{\bool}{{\{0, 1\}}}

\DeclareMathOperator{\questimator}{\normalfont{\textproc{QuantumEstimator}}}
\DeclareMathOperator{\clastimator}{\normalfont{\textproc{ClassicalEstimator}}}
\newcommand{\PM}[1]{#1\normalfont{\text{PM}}}
\DeclareMathOperator{\aPM}{\PM{\alpha}}

\newtheorem{theorem}{Theorem}
\newtheorem*{theorem*}{Theorem}
\newtheorem{lemma}[theorem]{Lemma}
\newtheorem{definition}[theorem]{Definition}
\newtheorem*{definition*}{Definition}
\newtheorem*{lemma*}{Lemma}

\newtheorem*{corollary*}{Corollary}

\newtheorem*{claim*}{Claim}

{\makeatletter
	\gdef\xxxmark{%
		\expandafter\ifx\csname @mpargs\endcsname\relax 
		\expandafter\ifx\csname @captype\endcsname\relax 
		\marginpar{xxx}
		\else
		xxx 
		\fi
		\else
		xxx 
		\fi}
	\gdef\xxx{\@ifnextchar[\xxx@lab\xxx@nolab}
	\long\gdef\xxx@lab[#1]#2{{\bf [\xxxmark #2 ---{\sc #1}]}}
	\long\gdef\xxx@nolab#1{{\bf [\xxxmark #1]}}
}

\title{A Quantum Advantage for a Natural Streaming Problem}
\date{}
\author{John Kallaugher\\The University of Texas at Austin\\
\texttt{jmgk@cs.utexas.edu}}

\begin{document}
\maketitle

\begin{abstract}
\noindent
Data streaming, in which a large dataset is received as a ``stream'' of
updates, is an important model in the study of space-bounded computation.
Starting with the work of Le Gall \lbrack SPAA~`06\rbrack, it has been known
that \emph{quantum} streaming algorithms can use asymptotically less space than
their classical counterparts for certain problems. However, so far, all known
examples of quantum advantages in streaming are for problems that are either
specially constructed for that purpose, or require many streaming passes over
the input.

We give a one-pass quantum streaming algorithm for one of the best-studied
problems in classical graph streaming---the \emph{triangle counting} problem.
Almost-tight parametrized upper and lower bounds are known for this problem in
the classical setting; our algorithm uses polynomially less space in certain
regions of the parameter space, resolving a question posed by Jain and Nayak in
2014 on achieving quantum advantages for natural streaming problems.
\end{abstract}

\section{Introduction}
\subsection{Streaming Algorithms}
Streaming algorithms are a class of algorithms for processing \emph{very large}
datasets that arrive ``one piece at a time''---some dataset too large to fit
into memory is built up by a series of updates. More formally, a vector $x \in
\Zbb^N$ is received as a series of updates $(\sigma_t)_{t = 1 \dots}$, where
each update $\sigma_t = ze_i$ consists of adding a scalar $z \in \Zbb$ to a
co-ordinate $i \in \brac{N}$, and the goal of a streaming algorithm is to
estimate some statistic of $x = \sum_{t}\sigma_t$ in $\lO{n}$
space\footnote{Space has been the primary object of study in the theory of
streaming algorithms. Update and pre- and post-processing time are typically,
although not necessarily, manageable if the space required by the algorithm is
small.}.

Streaming algorithms have been studied for a wide variety of problems, such as
cardinality estimation~\cite{FM85}, approximating the moments of a
vector~\cite{AMS96}, and subgraph counting~\cite{BKS02}. In this paper we will
be concerned with \emph{quantum} streaming algorithms.

\paragraph{Quantum Streaming} The prospect of space-constrained quantum
computers has motivated the study of \emph{quantum} streaming, in which a
stream of updates is received by an algorithm that is able to maintain a
quantum state and perform quantum operations (including measurements) on this
state as it processes the stream\footnote{Note that, despite the deferred
measurement principle, we may want to perform measurements between updates, as
the measurements we make may depend on the updates we see, and therefore we
cannot automatically push them to the end of the stream.}.

The study of quantum streaming algorithms started with~\cite{LG06}, in which a
problem was constructed that exhibits an exponential separation between quantum
and classical space complexity. The problem in question is not quite a
streaming problem in the sense we defined above, as the function tested depends
on the order of the updates of the stream, but in~\cite{GKKRW07} it was shown
that such a separation exists for an update-order-independent function.

This suggests the question, raised in~\cite{JN14}, of whether it is possible to
obtain such separations for ``natural'' problems. They proposed as a candidate
the problem of recognizing the $\texttt{Dyck}(2)$ language in the
stream\footnote{They use the broader definition of streaming that encompasses
update-order-dependent functions, but any separation for a function that does
\emph{not} depend on the order would also be one in that model.}, but while
better \emph{lower} bounds for this problem have since been shown~\cite{NT17},
better-than-classical upper bounds are still unknown. 

When many passes are allowed over the stream,~\cite{M16,YM19} demonstrate a
quantum advantage for the problem of estimating the frequency moments of a
vector, giving algorithms for various settings of the problem that can save a
$k^2$ factor in space complexity when they make $k$ passes, instead of the $k$
factor possible in classical streaming~\cite{AMS96,CR11}. But typically in
streaming the objective is to make only one pass over the stream, or at most
$\bO{1}$ passes. 

We resolve the question of~\cite{JN14}, giving a new one-pass streaming
algorithm for the \emph{triangle counting} problem, one of the best-studied
problems in graph streaming. 

\paragraph{Triangle Counting} In the (insertion-only) graph streaming model a
graph $G = (V, E)$ is received as a sequence of edges\footnote{In our vector
model described earlier, this corresponds to receiving the adjacency matrix of
the graph as a series of \emph{positive} updates to individual co-ordinates.
Other models of graph streaming, in which edges can be deleted as well as
added, also exist.} $(\sigma_t)_{t=1}^m$ from $E$.  The first problem to be
studied in this setting was that of estimating the number of triangles
(three-cliques) in $G$~\cite{BKS02}.

All algorithms for this problem are parametrized, as counting triangles
requires $\bOm{n^2}$ space if the number of them is sufficiently
small~\cite{BKS02}, and even graphs with $\bOm{m}$ triangles can be hard to
distinguish from triangle-free graphs in sufficiently ``hard''
graphs~\cite{BOV13}. The space complexity of such algorithms is therefore
typically quoted in terms of these parameters---it will often be unreasonable
to assume that the algorithm knows these parameters exactly in advance, but
constant factor bounds on them will suffice.

The best-known classical algorithm in this setting is from~\cite{JK21}, which
gives an \[
\bOt{\frac{m}{T}\cdot\paren*{\Delta_E + \sqrt{\Delta_V}}
\cdot\frac{1}{\varepsilon^2}\log \frac{1}{\delta}}
\]
space upper bound for obtaining a $(1 \pm \varepsilon)$-multiplicative
approximation with probability\footnote{For the remainder of this discussion we
will assume $\varepsilon, \delta$ are constant. Most algorithms for this
problem have a $\frac{1}{\varepsilon^2}\log \frac{1}{\delta}$ dependence, which
comes from taking the average of $\bT{\frac{1}{\varepsilon^2}}$
constant-variance estimators to obtain a $(1 \pm \varepsilon)$-multiplicative
approximation with 2/3 probability, then repeating $\bT{\log \frac{1}{\delta}}$
times and taking the median in order to amplify the success probability to $1 -
\delta$.} $1 - \delta$  in a graph with $m$ edges, $T$ triangles, and in which
no more than $\Delta_E$ triangles share an edge and no more than $\Delta_V$
share a vertex.  This algorithm is known to be classically optimal for this
parametrization, up to log factors, as~\cite{BOV13} gives a
$\bOm{\frac{m\Delta_E}{T}}$ lower bound and~\cite{KP17} gives
$\bOm{\frac{m\sqrt{\Delta_V}}{T}}$ when $T = \bO{m}$.

These two lower bounds are both based on reductions from communication
complexity. The first is from the Indexing problem~\cite{KNR95}, which is as
hard for quantum communication as it is for classical communication
(see~\cite{ANTV02}, in which it is called the problem of quantum random
access codes) and so the bound directly applies to any quantum streaming
algorithm. However, the second is from the Boolean Hidden Matching problem, in
particular the variant studied in~\cite{GKKRW07}. This problem \emph{is} easier
in quantum communication, and indeed was already used to prove a
quantum-classical streaming separation.

We give a quantum triangle counting algorithm that beats the classical lower
bound.
\begin{restatable}{theorem}{main}
\label{thm:main}
For any $\varepsilon, \delta \in (0,1\rbrack$, there is a quantum streaming
algorithm that uses \[
O\paren*{ \frac{m^{8/5}}{T^{6/5}}\Delta_E^{4/5}\log n
\cdot\frac{1}{\varepsilon^2}\log \frac{1}{\delta}}
\] quantum and classical bits in expectation to return a $(1 \pm
\varepsilon)$-multiplicative approximation to the triangle count in an
insertion-only graph stream with probability $1 - \delta$.

$m$ is the number of edges in the stream, $T$ the number of triangles, and
$\Delta_E$ the greatest number of triangles sharing any given edge.
\end{restatable}
In particular, this means that when $\Delta_E = \bO{1}$, $\Delta_V = \Omega(T)
= \Omega(m)$ (i.e. maximizing the separation, as $T$ must be $O(m)$ if the
classical lower bounds are to hold), we require $\bOt{m^{2/5}}$ space instead
of the $\bOm{\sqrt{m}}$ required by any classical algorithm\footnote{It remains
open whether this is the best separation possible---the
$\bOm{\frac{m\Delta_E}{T}}$ lower bound from Indexing disappears with these
parameter settings and so it is possible that even exponential advantages can
be achieved.}.
\subsection{Other Related Work}
Other work has investigated streaming problems with quantum
\emph{inputs}~\cite{BCG13,Y20}, as well as quantum versions of models closely
related to streaming, such as online algorithms~\cite{KKM18}, limited-width
branching programs~\cite{NHK00,SS05,AAKV18,HMWW20}, and finite
automata~\cite{KW97,AF98,MC00,ANTV02}.

Of particular relevance to streaming (see, e.g.~\cite{BGW20}) is the \emph{coin
problem}, in which a coin is flipped repeatedly, and the task is to determine
whether it is $p$-biased or $(p + \varepsilon)$-biased. This problem actually
exhibits \emph{arbitrary} quantum advantage, as any classical algorithm requires
at least $\log(p(1-p)) + \log(1/\varepsilon)$ space, while a quantum algorithm
can solve it with a single qubit~\cite{AD11}, although not if $\varepsilon$ is
unknown~\cite{KO17}.

\section{Overview of the Algorithm}
\subsection{Classical Triangle Counting}
To understand how maintaining a quantum state will help us with triangle
counting, we will start by describing an optimal \emph{classical} algorithm,
from~\cite{JK21}. For ease of exposition we will consider the related problem
of \emph{triangle distinguishing}---determining whether a graph $G = (V,E)$ has
$0$ triangles or whether it has at least $T$ triangles. As is often (although
not necessarily) the case with subgraph counting problems, converting to a
counting algorithm will be almost immediate.

\begin{samepage}
\noindent
The classical algorithm is as follows:
\begin{enumerate}
\item Sample vertices with probability $p$.
\item Sample edges incident to sampled vertices with probability $q$.
\item Whenever an edge arrives that ``completes'' a pair of sampled edges
incident to some sampled vertex (``a wedge centered at a sampled vertex'') we
will note that we have found a triangle.
\end{enumerate}
\end{samepage}
For any triangle in $G$, we will find it iff we sample its ``first'' vertex,
the vertex shared by the first two of its edges to arrive, and then sample both
of those edges. We can then distinguish between graphs with $0$ triangles and
those with $T$ by reporting whether we found a triangle or not (for counting,
we count the number of triangles we found and scale it by $p^{-1}q^{-2}$). The
vertex sampling can be implemented with a pairwise independent hash function,
so the expected space needed is $\bO{pqm\log n}$ for a graph with $m$ edges and
$n$ vertices, as each edge is kept with probability $pq$. 

How small can $pq$ be? This is given by the graph parameters $T$ (the number of
triangles in the graph), $\Delta_E$ (the maximum number of triangles sharing an
edge), and $\Delta_V$ (the maximum number of triangles sharing a vertex).
Ignoring constant factors, if we want to find a triangle with constant
probability:
\begin{itemize}
\item $p$ must be at least $\Delta_V/T$, as the triangles might share as few as
$T/\Delta_V$ ``first'' vertices.
\item $pq$ must be at least $\Delta_E/T$, as there might be a set $S$,
containing as few as $T/\Delta_E$ edges, such that every triangle has an edge
from $S$ as one of its first two edges.
\item $pq^2$ must be at least $1/T$, as each triangle will be found with probability $pq^2$.
\end{itemize}
It turns out that these are also sufficient, and subject to them $pq$ is minimized when \[
p = \frac{\Delta_V}{T}, \quad \quad \quad \quad \quad q =
\max\set*{\frac{\Delta_E}{\Delta_V}, \frac{1}{\sqrt{\Delta_V}}}
\]
giving an algorithm that uses \[
\bO{\frac{m}{T}\paren*{\Delta_E + \sqrt{\Delta_V}} \log n}
\]
space. Lower bounds from~\cite{BOV13, KP17} establish that (up to log factors)
both the $\Delta_E$ and $\Delta_V$ terms are necessary. However, while the
first of these is based on a reduction from the Indexing problem, which is
known to require as much quantum communication as classical communication to
solve~\cite{ANTV02}, the latter is based on a reduction from the Boolean Hidden Matching
problem (in particular, the ``$\alpha$-Partial Matching'' variant $\aPM_n$
of~\cite{GKKRW07}), which is known to exhibit an exponential separation between
classical and quantum communication.

\subsection{Quantum Triangle Counting}
\subsubsection{Two Players}
In seeking a quantum advantage, we consider the hard instance of~\cite{KP17},
depicted in Figure~\ref{fig:clashard}.  This is as follows:
\begin{enumerate}
\item $T/\Delta_V$ stars of degree $m\Delta_V/T$ arrive. We call the central
vertices of these ``hubs'' and their neighboring vertices ``spokes''.
\item Another $m$ edges arrive, all disjoint from each other. We have that
either, for each hub, $\Delta_V$ of these edges form triangles by connecting
two spokes of the hub, or none of them do for any hub.
\end{enumerate}
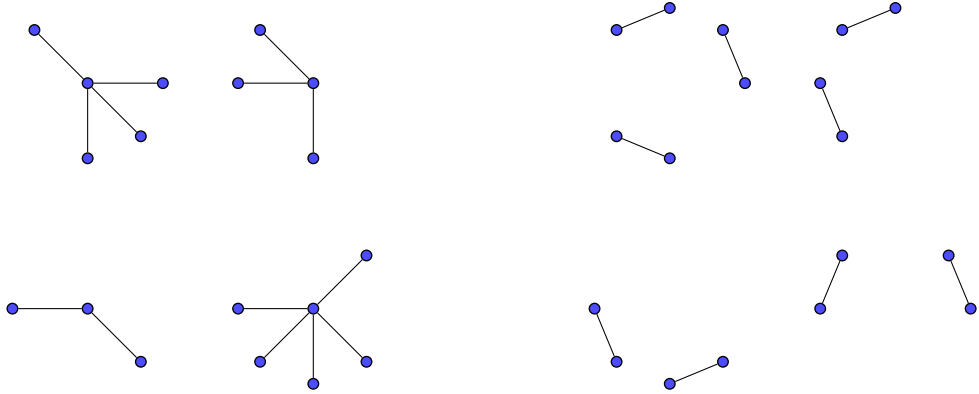
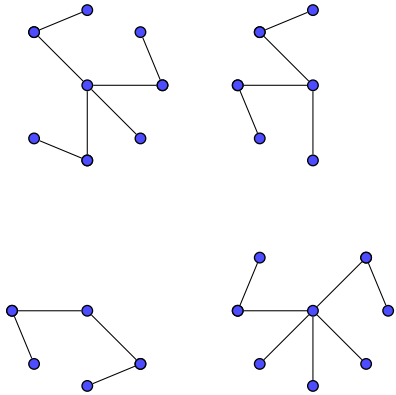
\begin{figure}
    \centering
    \begin{subfigure}[t]{0.47\textwidth}
        \centering
        \tikzset{basevertex/.style={shape=circle, line width=0.5,
        minimum size=4pt, inner sep=0pt, draw}}
        \tikzset{defaultvertex/.style={basevertex, fill=blue!70}}
        \begin{tikzpicture}[%
                VertexStyle/.style={defaultvertex},
                fat arrow/.style={single arrow,
                    thick,draw=blue!70,fill=blue!30,
                minimum height=10mm},
            scale = 1]

            \node[style={defaultvertex}](h) at (0,0) {};
            \foreach \i in {4,7} {
                \node[style={defaultvertex}](s\i) at ({1*cos(45 *
                \i)},{1*sin(45 *\i)})
                {};
                \draw (h) -- (s\i);
            }
            \node[style={defaultvertex}](h) at (3,0) {};
            \foreach \i in {1,4,5,6,7} {
                \node[style={defaultvertex}](s\i) at ({1*cos(45 *
                \i)+3},{1*sin(45 *\i)})
                {};
                \draw (h) -- (s\i);
            }
            \node[style={defaultvertex}](h) at (0,3) {};
            \foreach \i in {0,3,6,7} {
                \node[style={defaultvertex}](s\i) at ({1*cos(45 *
                \i)},{1*sin(45 *\i)+3})
                {};
                \draw (h) -- (s\i);
            }
            \node[style={defaultvertex}](h) at (3,3) {};
            \foreach \i in {3,4,6} { \node[style={defaultvertex}](s\i) at
            ({1*cos(45 * \i)+3},{1*sin(45 *\i)+3})
            {};
            \draw (h) -- (s\i);
        }
    \end{tikzpicture}
    \caption{The first half of the stream is $T/\Delta_V$ ``hubs'', each a
    vertex with edges to $m\Delta_V/T$ ``spokes''.}
\end{subfigure}
\begin{subfigure}[t]{0.47\textwidth}
    \centering
    \tikzset{basevertex/.style={shape=circle, line width=0.5,
    minimum size=4pt, inner sep=0pt, draw}}
    \tikzset{defaultvertex/.style={basevertex, fill=blue!70}}
    \begin{tikzpicture}[%
            VertexStyle/.style={defaultvertex},
            fat arrow/.style={single arrow,
                thick,draw=blue!70,fill=blue!30,
            minimum height=10mm},
        scale = 1]
        \foreach \i in {4,5,6,7} {
            \node[style={defaultvertex}](s\i) at ({5 + 1*cos(45 *
            \i)},{1*sin(45
            *\i)}) {};
        }
        \draw (s4) -- (s5);
        \draw (s6) -- (s7);
        \foreach \i in {0,1,3,4} {
            \node[style={defaultvertex}](s\i) at ({5 + 1*cos(45 *
            \i)+3},{1*sin(45 *\i)}) {};
        }
        \draw (s0) -- (s1);
        \draw (s3) -- (s4);
        \foreach \i in {0,1,2,3,5,6} {
            \node[style={defaultvertex}](s\i) at ({5 + 1*cos(45 *
            \i)},{1*sin(45 *\i)+3}) {};
        }
        \draw (s0) -- (s1);
        \draw (s2) -- (s3);
        \draw (s5) -- (s6);
        \foreach \i in {2,3,4,5} {
            \node[style={defaultvertex}](s\i) at ({5 + 1*cos(45 *
            \i)+3},{1*sin(45 *\i)+3}) {};
        }
        \draw (s2) -- (s3);
        \draw (s4) -- (s5);
    \end{tikzpicture}
    \caption{They are followed by $m$ edges, disjoint from each other but
    potentially incident to the ``spoke'' vertices of the hubs.}
\end{subfigure}
    \begin{subfigure}[t]{\textwidth}
        \centering
        \tikzset{basevertex/.style={shape=circle, line width=0.5,
        minimum size=4pt, inner sep=0pt, draw}}
        \tikzset{defaultvertex/.style={basevertex, fill=blue!70}}
        \begin{tikzpicture}[%
                VertexStyle/.style={defaultvertex},
                fat arrow/.style={single arrow,
                    thick,draw=blue!70,fill=blue!30,
                minimum height=10mm},
            scale = 1]

            \node[style={defaultvertex}](h) at (0,0) {};
            \foreach \i in {4,7} {
                \node[style={defaultvertex}](s\i) at ({1*cos(45 *
                \i)},{1*sin(45 *\i)})
                {};
                \draw (h) -- (s\i);
            }
            \node[style={defaultvertex}](h) at (3,0) {};
            \foreach \i in {1,4,5,6,7} {
                \node[style={defaultvertex}](s\i) at ({1*cos(45 *
                \i)+3},{1*sin(45 *\i)})
                {};
                \draw (h) -- (s\i);
            }
            \node[style={defaultvertex}](h) at (0,3) {};
            \foreach \i in {0,3,6,7} {
                \node[style={defaultvertex}](s\i) at ({1*cos(45 *
                \i)},{1*sin(45 *\i)+3})
                {};
                \draw (h) -- (s\i);
            }
            \node[style={defaultvertex}](h) at (3,3) {};
            \foreach \i in {3,4,6} { \node[style={defaultvertex}](s\i) at
            ({1*cos(45 * \i)+3},{1*sin(45 *\i)+3})
            {};
            \draw (h) -- (s\i);
        }
        \foreach \i in {4,5,6,7} {
            \node[style={defaultvertex}](s\i) at ({1*cos(45 *
            \i)},{1*sin(45
            *\i)}) {};
        }
        \draw (s4) -- (s5);
        \draw (s6) -- (s7);
        \foreach \i in {0,1,3,4} {
            \node[style={defaultvertex}](s\i) at ({1*cos(45 *
            \i)+3},{1*sin(45 *\i)}) {};
        }
        \draw (s0) -- (s1);
        \draw (s3) -- (s4);
        \foreach \i in {0,1,2,3,5,6} {
            \node[style={defaultvertex}](s\i) at ({1*cos(45 *
            \i)},{1*sin(45 *\i)+3}) {};
        }
        \draw (s0) -- (s1);
        \draw (s2) -- (s3);
        \draw (s5) -- (s6);
        \foreach \i in {2,3,4,5} {
            \node[style={defaultvertex}](s\i) at ({1*cos(45 *
            \i)+3},{1*sin(45 *\i)+3}) {};
        }
        \draw (s2) -- (s3);
        \draw (s4) -- (s5);
    \end{tikzpicture}
    \caption{The union of these two halves may contain $0$ or $T$ triangles.}
\end{subfigure}\caption{A hard graph for classical algorithms, when $\Delta_E = 1$ but
$\Delta_V$ is unrestricted.}
\label{fig:clashard}
\end{figure}

We will start by considering this in the simpler\footnote{Any streaming
algorithm for this problem immediately gives a one-way protocol for the
two-player version, with message size equal to the space needed by the
algorithm. Alice can run the streaming algorithm on her input, send the
algorithm's state to Bob, and then he can initialize it with that state and run
it on his input.} \emph{two-player} setting---Alice gets the first half of the
stream, Bob gets the second half of the stream, and Alice wants to send Bob a
message that he can use with his input to determine whether the graph has 0 or
$T$ triangles.

If she wanted to do this by sending some subset of the edges, she would need to
send $\bO{\frac{m\sqrt{\Delta}}{T}}$ of them---within any given hub the edges
are indistinguishable, so at best she can choose one specific hub and send a
$1/\sqrt{\Delta_V}$ fraction of its edges, to have a $(1/\sqrt{\Delta_V})^2 =
1/\Delta_V$ chance of finding any given one of its $\Delta_V$
triangles\footnote{This is essentially identical to running the algorithm
of~\cite{JK21} on the input.}.  By embedding an instance of
$\PM{(T/m)}_{m\Delta_V/T}$ in a hub, and then copying that hub $T/\Delta_V$
times, it can be shown that there is no asymptotically better classical message
Alice can send.

What if Alice is allowed to send a quantum message? We cannot emulate the
$\aPM$ protocol of~\cite{GKKRW07} directly, as not all graphs of this form will
correspond to an embedding of $\aPM$. Instead, if the set of hub vertices is $H$,
with spoke vertices $S_u$ for each $u \in H$, Alice may construct the $O(\log
m)$-bit quantum state \[
\frac{1}{\sqrt{m}}\sum_{u \in H}\sum_{v \in S_u}\ket{\dedge{uv}}
\]
where $\dedge{uv}$ denotes the \emph{directed} edge from\footnote{Of course, in
the general case, we won't know which vertices are hubs and which are spokes,
so each edge $uv$ will need to be included as both $\dedge{uv}$ and
$\dedge{vu}$.} $u$ to $v$. As Bob's edges are disjoint, he can then construct
an orthonormal basis of $\Cbb^{\abs{V^2}}$ which contains \[
\frac{\ket{\dedge{wu}} + \ket{\dedge{wv}}}{\sqrt{2}}, \frac{\ket{\dedge{wu}} -
\ket{\dedge{wv}}}{\sqrt{2}}
\]
for every $w \in V$ and edge $uv$ in his set of $m$ edges. If he measures
Alice's state in this basis, he will see:
\begin{itemize}
\item Each basis element of the form $\frac{\ket{\dedge{wu}} +
\ket{\dedge{wv}}}{\sqrt{2}}$ with probability $1/2m$ if Alice has a hub $w$
with exactly one of the spokes $u, v$, and $2/m$ if Alice has a hub $w$ with
\emph{both} of these as spokes (i.e.\ if $uv$ completes a triangle in Alice's
input).
\item Each basis element of the form $\frac{\ket{\dedge{wu}} - \ket{\dedge{wv}}}{\sqrt{2}}$ with probability $1/2m$ if Alice has a hub $w$
with exactly one of the spokes $u, v$, and 0 if Alice has a hub $w$ with both
of these as spokes.
\end{itemize}
So if $G$ is triangle-free these are the same, but if $G$ has $T$ triangles they differ by $2T/m$, and so Bob can work out which situation they are in if Alice sends him $\Theta(m^2/T^2)$ copies of this state, at the cost of $\bT{\frac{m^2}{T^2}\log m}$ qubits.

To generalize this technique, we will need to address two questions: how to
construct and measure the state in the stream, and what to do when the edges we
want to measure by are \emph{not} disjoint, and therefore do not give an
orthonormal basis.
\subsubsection{Streaming}
\paragraph{Constructing the State} We start by constructing the superposition \[
\frac{1}{\sqrt{2m}}\sum_{i=1}^{2m} \ket{i}
\]
of $2m$ ``dummy'' states. Then, whenever we process the $i^\text{th}$ edge $uv$
in the stream, we swap the dummy states $\ket{2i-1}$, $\ket{2i}$ for
$\ket{\dedge{uv}}$, $\ket{\dedge{vu}}$.

\paragraph{Measurements} If we could remember all of the measurements we want
to make, we could take this state and perform the measurements we made in the
two-player game at the end of the stream. However, this would require
remembering every edge we see, so instead we perform the measurements one edge
at a time.

When the edge $uv$ arrives, we construct the set of projectors $\Oc^{uv} =
\set{O^{uv}_b : b \in \set{-1,0,1}}$ given by
\begin{align*}
O^{uv}_1 &= \frac{1}{2}\sum_{w \in V} (\ket{\dedge{wu}} + \ket{\dedge{wv}})(\bra{\dedge{wu}} + \bra{\dedge{wv}})\\
O^{uv}_{-1} &= \frac{1}{2}\sum_{w \in V} (\ket{\dedge{wu}} - \ket{\dedge{wv}})(\bra{\dedge{wu}} - \bra{\dedge{wv}})\\
O^{uv}_0 &= I - \sum_{w \in V}(\ket{\dedge{wu}}\bra{\dedge{wu}} + \ket{\dedge{wv}}\bra{\dedge{wv}})
\end{align*}
and measure with them\footnote{If performing all the measurements at the end of
the stream is desired, an alternative with the same outcome (and using only a
constant factor extra qubits) is to perform the unitary swapping
$\frac{\ket{\dedge{wu}} + \ket{\dedge{wv}}}{2}$ with $\ket{\dedge{wuv}+}$ and
$\frac{\ket{\dedge{wu}} - \ket{\dedge{wv}}}{2}$ with $\ket{\dedge{wuv}-}$ for
each $w$, and then measuring in the standard basis at the end of the stream.},
with $b$ the outcome corresponding to the operator $O_b^{uv}$. If we see $+1$
or $-1$ we terminate the algorithm and return with that value, while if we see
$0$ we continue processing the
stream.

This gives us a somewhat different result from performing the measurement at
the end of the stream, since now when we measure by an edge $uv$ we can only
pick up on triangles $wuv$ such that $wu$ and $wv$ appear before $uv$. However,
this is not a concern, as for each triangle there will be exactly one final
edge that arrives.

\paragraph{Non-Disjoint Edges} However, there is a problem with this strategy.
The two-player strategy worked because the edges in Bob's set $B$ were
\emph{disjoint}, and so the elements \[
\set*{\frac{\ket{\dedge{wu}} + \ket{\dedge{wv}}}{\sqrt{2}},
\frac{\ket{\dedge{wu}} - \ket{\dedge{wv}}}{\sqrt{2}} : w \in V, uv \in B}
\]
could be a subset of an orthonormal basis of $\Rbb^{\abs{V^2}}$. But in the
general setting, we have to measure by every edge, as we do not know in advance
which edges will be ``last edges'' of triangles. 

Now, the individual edge measurements described in the previous section are
still valid, but what will be the impact of measuring by them? After measuring
with $\Oc^{uv}$, if we \emph{do not} terminate the algorithm, then the state
has been projected onto $O^{uv}_0$ (and re-weighted appropriately). This means
that for all $w \in V$, if either $\ket{\dedge{wu}}$ or $\ket{\dedge{wv}}$ were
present in the superposition, they will now be gone.

The consequence of this is that instead of an estimator of the number of
triangles in the graph, we instead have an estimator of the number of triangles
$wuv$ such that $wu$, $wv$ appear in the stream and then no edges incident to
$u$ or $v$ arrive before $uv$ does. However, this is not very useful, as this
could easily be $0$ even in a graph with many triangles.

In order to have at least some chance of finding triangles that do not fit this
description, we will refrain from measuring with \emph{all} of the edges we
see. Suppose that when seeing the edge $uv$ we only perform the $\Oc^{uv}$
measurement with probability $1/k$, for some $k$. Then, for every triangle
$wuv$, if its edges arrive in the order $wu$, $wv$, $uv$, the probability that,
$\ket{\dedge{wu}}$ and $\ket{\dedge{wv}}$ are still in the state when $uv$
arrives is \[
(1 - 1/k)^{\degb{wuv} + \degb{wvu}}
\]
where $\degb{wuv}$ denotes the degree between $wu$ and $uv$, the number of
edges incident to $u$ that arrive in between $wu$ and $uv$. So as the
measurement $\Oc^{uv}$ itself is performed with probability $1/k$, the
probability that we perform it \emph{and} both $\ket{\dedge{wu}}$ and
$\ket{\dedge{wv}}$ are in the state at the time is \[
\frac{1}{k}(1 - 1/k)^{\degb{wuv} + \degb{wvu}}\text{.}
\]

This means that any triangle with $\degb{wuv} + \degb{wvu} \le k$ has at least
an $\bOm{1/k}$ probability of contributing to our estimator. More specifically,
we can estimate $\Tlk$, given by down-weighting every triangle by $(1 -
1/k)^{\degb{wvu} + \degb{wuv}}$, by running $\bT{\paren{km/\Tlk}^2}$ copies of
the estimator in parallel, as now the probability of seeing $+1$ is only
$m/k\Tlk$ greater than the probability of seeing $-1$. If we only care to
estimate it to $\varepsilon T$ accuracy for some constant $\varepsilon$ (i.e.\
we are fine having a poor multiplicative estimate when $\Tlk$ is small) we can
replace this with $\bT{\paren{km/T}^2}$ copies instead.

However, this is not helpful on its own, as it is possible to have a graph
stream where $\Tlk$ is much smaller than $T$ for any $k \ll m$. Consider the
following stream, depicted in Figure~\ref{fig:quahard}:
\begin{enumerate}
\item The edges $(uw_i)_{i=1}^T$ and $(w_iv_i)_{i=1}^T$ arrive for some
sequences of unique vertices $(w_i)_{i=1}^m$, $(v_i)_{i=1}^m$ and some unique
fixed vertex $u$.
\item The edges $(uz_i)_{i=1}^m$ arrive for some sequence of unique vertices
$(z_i)_{i=1}^m$.
\item The edges $(uv_i)_{i=1}^T$ arrive.
\end{enumerate}
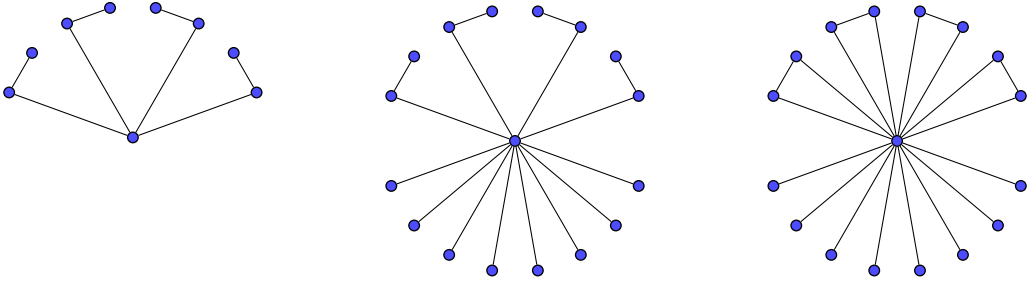
\begin{figure}
    \centering
    \begin{subfigure}[t]{0.30\textwidth}
        \centering
        \tikzset{basevertex/.style={shape=circle, line width=0.5,
        minimum size=4pt, inner sep=0pt, draw}}
        \tikzset{defaultvertex/.style={basevertex, fill=blue!70}}
        \begin{tikzpicture}[%
                VertexStyle/.style={defaultvertex},
                fat arrow/.style={single arrow,
                    thick,draw=blue!70,fill=blue!30,
                minimum height=10mm},
            scale = 1]

            \node[style={defaultvertex}](h) at (0,0) {};
            \foreach \i in {1,3} {
                \node[style={defaultvertex}](s\i) at ({1.75*cos(20 *
                \i)},{1.75*sin(20 *\i)})
                {};
                \node[style={defaultvertex}](s{\i+1}) at ({1.75*cos(20 *
                (\i+1))},{1.75*sin(20 *(\i+1))})
                {};
                \draw (h) -- (s\i);
                \draw (s\i) -- (s{\i+1});
            }
            \foreach \i in {6,8} {
                \node[style={defaultvertex}](s\i) at ({1.75*cos(20 *
                \i)},{1.75*sin(20 *\i)})
                {};
                \node[style={defaultvertex}](s{\i-1}) at ({1.75*cos(20 *
                (\i-1))},{1.75*sin(20 *(\i-1))})
                {};
                \draw (h) -- (s\i);
                \draw (s\i) -- (s{\i-1});
            }
            \foreach \i in {10,...,17} {
                \node(s\i) at ({1.75*cos(20 *
                \i)},{1.75*sin(20 *\i)})
                {};
            }
        \end{tikzpicture}
        \caption{$T$ wedges $(uw_iv_i)_{i=1}^T$ arrive, all incident to $u$.}
    \end{subfigure}
    \begin{subfigure}[t]{0.30\textwidth}
        \centering
        \tikzset{basevertex/.style={shape=circle, line width=0.5,
        minimum size=4pt, inner sep=0pt, draw}}
        \tikzset{defaultvertex/.style={basevertex, fill=blue!70}}
        \begin{tikzpicture}[%
                VertexStyle/.style={defaultvertex},
                fat arrow/.style={single arrow,
                    thick,draw=blue!70,fill=blue!30,
                minimum height=10mm},
            scale = 1]

            \node[style={defaultvertex}](h) at (0,0) {};
            \foreach \i in {1,3} {
                \node[style={defaultvertex}](s\i) at ({1.75*cos(20 *
                \i)},{1.75*sin(20 *\i)})
                {};
                \node[style={defaultvertex}](s{\i+1}) at ({1.75*cos(20 *
                (\i+1))},{1.75*sin(20 *(\i+1))})
                {};
                \draw (h) -- (s\i);
                \draw (s\i) -- (s{\i+1});
            }
            \foreach \i in {6,8} {
                \node[style={defaultvertex}](s\i) at ({1.75*cos(20 *
                \i)},{1.75*sin(20 *\i)})
                {};
                \node[style={defaultvertex}](s{\i-1}) at ({1.75*cos(20 *
                (\i-1))},{1.75*sin(20 *(\i-1))})
                {};
                \draw (h) -- (s\i);
                \draw (s\i) -- (s{\i-1});
            }
            \foreach \i in {10,...,17} {
                \node[style={defaultvertex}](s\i) at ({1.75*cos(20 *
                \i)},{1.75*sin(20 *\i)})
                {};
                \draw (h) -- (s\i);
            }
        \end{tikzpicture}
        \caption{$m$ more edges $(uz_i)_{i=1}^m$ arrive incident to $u$.}
    \end{subfigure}
    \begin{subfigure}[t]{0.30\textwidth}
        \centering
        \tikzset{basevertex/.style={shape=circle, line width=0.5,
        minimum size=4pt, inner sep=0pt, draw}}
        \tikzset{defaultvertex/.style={basevertex, fill=blue!70}}
        \begin{tikzpicture}[%
                VertexStyle/.style={defaultvertex},
                fat arrow/.style={single arrow,
                    thick,draw=blue!70,fill=blue!30,
                minimum height=10mm},
            scale = 1]

            \node[style={defaultvertex}](h) at (0,0) {}; 
            \foreach \i in {1,3} { 
                \node[style={defaultvertex}](s\i) at ({1.75*cos(20 *
                \i)},{1.75*sin(20 *\i)})
                {};
                \node[style={defaultvertex}](s{\i+1}) at ({1.75*cos(20 *
                (\i+1))},{1.75*sin(20 *(\i+1))})
                {};
                \draw (h) -- (s\i);
                \draw (s\i) -- (s{\i+1});
                \draw (s{\i+1}) -- (h);
            }
            \foreach \i in {6,8} { 
                \node[style={defaultvertex}](s\i) at ({1.75*cos(20 *
                \i)},{1.75*sin(20 *\i)})
                {};
                \node[style={defaultvertex}](s{\i-1}) at ({1.75*cos(20 *
                (\i-1))},{1.75*sin(20 *(\i-1))})
                {};
                \draw (h) -- (s\i);
                \draw (s\i) -- (s{\i-1});
                \draw (s{\i-1}) -- (h);
            }
            \foreach \i in {10,...,17} {
                \node[style={defaultvertex}](s\i) at ({1.75*cos(20 *
                \i)},{1.75*sin(20 *\i)})
                {};
                \draw (h) -- (s\i);
            }
        \end{tikzpicture}
        \caption{$T$ triangles
$(w_iuv_vi)_{i=1}^T$ are formed, each with $\degb{w_iuv_i} = m$.}
    \end{subfigure}
    \caption{A hard graph for our quantum estimator, with
        $\degb{wuv} + \degb{wvu}$ large for every triangle $wuv$.}
        \label{fig:quahard}
    \end{figure}

Now every triangle $w_iuv_vi$ in this stream has $\degb{w_iuv_i} = m$, and so
$\Tlk$ will be very small if $k \ll m$.

However, even though this corresponds to a graph where $\Delta_V$ is large, it
will still be easy for the \emph{classical} algorithm we described at the start
of this section.

\paragraph{Hybrid Quantum-Classical Algorithm} In the stream described above,
the first two edges of each triangle $w_iuv_i$ are $w_iu$, $w_iv_i$, and so if
we run the classical algorithm from the start of this section with $q = 1$, the
triangle will be found if $w_i$ is sampled by the vertex sampling stage. As the
vertices $(w_i)_{i=1}^T$ are disjoint, the classical algorithm can succeed with
$p = \bT{1/T}$, for a classical space complexity of $\bO{m/T}$. 

More generally, if $\bOm{m/k}$ triangles $(wu_iv_i)_{i=1}^{m/k}$ share the same
``first'' vertex $w$, with distinct $(u_i)_{i=1}^{m/k}$,  $(v_i)_{i=1}^{m/k}$
then \[
\sum_{i=1}^{m/k} \paren*{\degb{wv_iu_i} + \degb{wu_iv_i}} \le
\sum_{i=1}^{m/k}\paren*{d_{u_i} + d_{v_i}} \le 2m
\]
and so the \emph{average} value of $ \paren*{\degb{wu_iv_i} + \degb{wv_iu_i}}$
across these triangles is at most $\bO{k}$. This means that, if we consider the
set of triangles $wuv$ with $\degb{wuv} + \degb{wvu} \ge k$, at most
$\bO{m/k}$ of them can share any one ``first'' vertex. 

This means that a classical algorithm can count them to $\varepsilon T$
accuracy (for some constant $\varepsilon$) with $s$ space by sampling vertices
with probability $p = m/T k$ and incident edges with probability
$\sqrt{k/m}$, for $\bO{m^{3/2}/T \sqrt{k}}$ total samples in expectation.
Moreover, by maintaining degree counters for the endpoints of each edge it
samples, such an algorithm can record $\degb{wuv} + \degb{wvu}$ for each
triangle $wuv$ it samples, and therefore estimate\footnote{Technically the
algorithm described here only allows estimating the sum of $1 - (1 -
1/k)^{\degb{wuv} + \degb{wvu}}$ over triangles where $\degb{wuv} + \degb{wvu}
\ge k$. But by analyzing the variance of the estimator more carefully it is
possible to replace this with the sum of that over \emph{all} triangles, i.e.\
$T - \Tlk$.} $\Tgk = T - \Tlk$.

In other words, triangles that are hard for the classical algorithm to count
are easier for the quantum algorithm to count, and vice-versa. This suggests
the following hybrid algorithm:
\begin{enumerate}
\item Choose $k$ appropriately.
\item Use the quantum algorithm to estimate $\Tlk$ to $\varepsilon T/2$ error.
\item Use the classical algorithm to estimate $\Tgk = T - \Tlk$ to $\varepsilon
T/2$ error.
\item Return the sum of the estimates.
\end{enumerate}
What should $k$ be? We want to minimize \[
\paren*{\frac{km}{T}}^2 + \frac{m^{3/2}}{T\sqrt{k}}\text{.}
\]
Setting $k = T^{2/5}/m^{1/5}$ gives us a \[
\bOt{\frac{m^{8/5}}{T^{6/5}}}
\]
space algorithm, that becomes \[
\bOt{\frac{m^{8/5}}{T^{6/5}} \Delta_E^{4/5}}
\]
when our bounds are modified to account for up to $\Delta_E$ triangles sharing
an edge. When e.g. $\Delta_E = O(1)$, $\Delta_V = \Omega(T)$ and $T =
\lOm{m^{6/7}}$, this is less space than the best possible classical algorithm.

\section{Preliminaries}
\subsection{General Notation}
$k \in \brac{0, m}$ is a parameter shared by the quantum and classical
algorithms, to be specified later.

Let $G = (V,E)$ be a graph on $n$ vertices, received as a stream of undirected
edges, adversarially ordered. Let $m \le \binom{n}{2}$ be an upper bound on
the number of edges in the graph (and thus the number of updates in the
stream). We write the stream $\sigma = (\sigma_i)_{i=1}^m$, for $\sigma \in E$.
We will write $\sigma^{\le t} = (\sigma_i)_{i=1}^t$. 

We will write $N(v)$ for the neighborhood of any $v \in V$, and $d_v$ for
$\abs{N(v)}$.

We will use $\dedge{uv}$ to denote a \emph{directed} edge from $u$ to $v$, and
$uv$ (or $vu$) to refer to the undirected edge (and so $uv = vu$ while
$\dedge{uv} \not= \dedge{vu}$). We will write $\dedge{E}$ for the set of
directed edges and $E$ for the set of edges.

We will use $\Ibb(p)$ to denote the indicator on whether the predicate $p$
holds.

\subsection{Triangles}
We use $T$ to refer to the number of triangles in $G$, $\Delta_E \ge 1$ to
refer to the maximum number of them sharing a single edge (or 1 if $G$ is
triangle-free).

Fix any ordering of the stream. For any edges $e, f$, we will write $e \bef f$
if $e$ arrives before $f$ in the stream. For any vertices $u,v,w \in V$ such
that $uv, vw \in E$ and $uv \bef vw$, let the degree between $uv$ and $vw$,
$\degb{uvw}$ be the number of edges incident to $v$ that arrive in between $uv$
and $vw$ (not including $uv$ or $vw$ themselves).

For any triple of vertices $(u,v,w) \in V^3$ let \[
\tlk_{uvw} = \begin{cases}
(1 - 1/k)^{\degb{uvw} + \degb{uwv}} &\mbox{if $\set{u, v, w}$ is a triangle in
the graph and $uv \bef uw \bef vw$}\\
0 &\mbox{otherwise.}
\end{cases}
\]
Likewise, let \[
\tgk_{uvw} = \begin{cases}
1 - (1 - 1/k)^{\degb{uvw} + \degb{uwv}} &\mbox{if $\set{u, v, w}$ is a triangle
in the graph and $uv \bef uw \bef vw$}\\
0 &\mbox{otherwise.}
\end{cases}
\]
We will write $\Tlk, \Tgk$ for $\sum_{(u,v,w) \in V^3} \tlk_{uvw}$, $\sum_{(u,v,w)
\in V^3} \tgk_{uvw}$, respectively, so that $T = \Tlk + \Tgk$.

For any vertex $u \in V$, we will write $\Tlk_{u} = \sum_{(v,w) \in V^2}
\tlk_{uvw}$ and $\Tgk_{u} = \sum_{(v,w) \in V^2} \tgk_{uvw}$, so $\sum_{u \in V}
\Tlk_u = \Tlk$ and $\sum_{u \in V} \Tgk_u = \Tgk$.

\section{Quantum Estimator}
Each instance of the quantum algorithm will maintain $\beta = 2\ceil{\log n} +
1$ qubits, indexing the set $\dedge{E} \cup \brac{2m}$. We will write the basis
states as $\ket{\dedge{uv}}$, $\ket{t}$ for $\dedge{uv} \in \dedge{E}$, $t \in
\brac{2m}$.

Let $f : \brac{m} \rightarrow \bool$ be a fully independent hash function such that \[
f(t) = \begin{cases}
1 & \mbox{with probability $1/k$}\\
0 & \mbox{otherwise.}
\end{cases}
\]
While $f$ is a fully random function, and so would be infeasible to store, our
algorithm will only need to query $f(t)$ at the time step $t$, and therefore
will not need to store it.

After the $t^\text{th}$ update, the algorithm will either terminate or maintain
the state \[
\Sigma_t = \frac{\sum_{i=2t+1}^{2m} \ket{i} + \sum_{\dedge{uv} \in S_t}
\ket{\dedge{uv}}}{\sqrt{2m - 2t + |S_t|}}
\]
where \[
S_t = \set{\dedge{uv} : \exists i \in \brac{t}, \sigma_i = uv, \forall j =  i+1,
\dots, t,  f(j) = 0 \vee v \not \in \sigma_j}.
\] 
That is, $S_t$ contains the directed edges $\dedge{uv}$ and $\dedge{vu}$ for
every edge $uv$ that has arrived at time $t$, except that whenever an edge $wz$
arrives at time $s$, if $f(s) = 1$ all edges directed towards either $w$ or $z$
are removed.

At each step $t$, the algorithm will first apply a unitary transformation,
depending on $t$ and the edge $\sigma_t$, to take $\Sigma_{t-1}$ to $\Sigma_t$,
and then, if $f(t) = 1$, it will measure $\Sigma_t$ with an operator depending
only on $\sigma_t$. We now define this transformation and measurement operator.

\begin{definition}
Take $\set{\ket{x} : x \in \dedge{E} \cup \brac{m}}$ and extend it to a basis
of $\Cbb^{\beta}$. For each $t \in \brac{m}$, $uv \in E$, if $u < v$ the
unitary transformation $U^t_{uv}$ is given by swapping the basis elements
$\ket{(2t-1)}$ and $\ket{\dedge{uv}}$, and swapping $\ket{2t}$ and
$\ket{\dedge{vu}}$. If $v < u$, $U^t_{uv} = U^t_{vu}$.
\end{definition}

\begin{definition}
For any $uv \in E$, the set of measurement operators $\Oc^{uv} = \set{O^{uv}_b
: b \in \set{-1,0,1}}$ is given by
\begin{align*}
O^{uv}_1 &= \frac{1}{2}\sum_{w \in V} (\ket{\dedge{wu}} +
\ket{\dedge{wv}})(\bra{\dedge{wu}} + \bra{\dedge{wv}})\\
O^{uv}_{-1} &= \frac{1}{2}\sum_{w \in V} (\ket{\dedge{wu}} -
\ket{\dedge{wv}})(\bra{\dedge{wu}} - \bra{\dedge{wv}})\\
O^{uv}_0 &= I - \sum_{w \in V}(\ket{\dedge{wu}}\bra{\dedge{wu}} +
\ket{\dedge{wv}}\bra{\dedge{wv}})
\end{align*}
with $b$ the outcome corresponding to operator $O^{uv}_b$.
\end{definition}
Note that $O^{uv}_1O^{uv}_{-1} = 0$, and $O^{uv}_0 = I - O^{uv}_1 -
O^{uv}_{-1}$, so this is a complete set of orthogonal projectors.

We can now define the algorithm.

\begin{algorithm}[H]
\caption{Quantum estimator for $\Tlk$}
\begin{algorithmic}[1]
\Procedure{QuantumEstimator}{$k$}
\State $t \gets 1$ 
\State $\Sigma_t \gets \frac{1}{\sqrt{m}}\sum_{i=1}^m \ket{i}$
\For{each update $uv$}
\State $\Sigma_t \gets U^t_{uv}\Sigma_{t-1}$
\If{$f(t) = 1$} \Comment{$f(t)$ is not re-used so we can generate it here.}
\State Measure $\Sigma_t$ with the operators $\Oc^{uv}$, storing the result in
$\bb$.
\If{$\bb \not= 0$}
\State \return $\bb$
\EndIf
\EndIf
\State $t \gets t + 1$
\EndFor
\State $\bb \gets 0$
\State \return $\bb$
\EndProcedure
\end{algorithmic}
\end{algorithm}

\begin{lemma}
\label{lm:stateinv}
For all $t = 0, \dots, m$, after $\questimator(k)$ has processed $t$ updates,
either it will have returned or \[
\Sigma_t = \frac{\sum_{i=2t+1}^{2m} \ket{i} + \sum_{\dedge{uv} \in S_t}
\ket{\dedge{uv}}}{\sqrt{2m - 2t + |S_t|}}
\]
where \[
S_t = \set{\dedge{uv} : \exists i \in \brac{t}, uv = \sigma_i, \forall j =
i+1, \dots, t,  f(j) = 0 \vee v \not \in \sigma_j}\text{.}
\]
\end{lemma}
\begin{proof}
We proceed by induction. For $t = 0$, \[
\Sigma_t = \frac{1}{\sqrt{m}}\sum_{i=1}^{2m}\ket{i}
\]
and so the result holds. Now, for any $t \in \brac{m-1}$, suppose that the result
holds after $t$ updates.  Let $xy$ (with $x<y$) be the $(t + 1)^\text{th}$
update. Then after applying the unitary $U_{xy}^{t+1}$, \[
\Sigma_{t+1} = \frac{\sum_{i=2t+3}^{2m} \ket{i} + \sum_{\dedge{uv} \in S_t \cup
\set {\dedge{xy}, \dedge{yx}}} \ket{\dedge{uv}}}{\sqrt{2m - 2t + |S_t|}}
\]
as $U^{t+1}_{xy}$ swapped the basis vectors $\ket{(2t+1)}$, $\ket{(2t+2)}$
for $\ket{\dedge{xy}}$, $\ket{\dedge{yx}}$.

So if $f(t+1) = 0$, the result continues to hold for $t+1$, as in this case \[
S_{t+1} = S_t \cup \set{\dedge{xy}, \dedge{yx}}
\]
and so $2m - 2t + \abs{S_t} = 2m - 2(t+1) + \abs{S_{t+1}}$. Now suppose $f(t+1)
= 1$. Let \[
W = \set{\dedge{wz} \in S_t : z \in \set{x,y}}\text{.}
\] Then \[
S_{t+1} = \set{\dedge{xy}, \dedge{yx}} \cup S_t \setminus W
\] and 
\begin{align*}
U^{t+1}_{xy}\Sigma_t &= \frac{\sum_{i=2t+3}^{2m} \ket{i} + \sum_{\dedge{uv} \in
S_t \cup \set {\dedge{xy}, \dedge{yx}}} \ket{\dedge{uv}}}{\sqrt{2m - 2t +
|S_t|}}\\
&= \frac{\sum_{i=2t+3}^{2m} \ket{i} + \sum_{\dedge{uv} \in S_{t+1}}
\ket{\dedge{uv}}}{\sqrt{2m - 2t + \abs{S_{t}}}} +
\frac{\sum_{\dedge{uv} \in W} \ket{\dedge{uv}}}{\sqrt{2m - 2t + |S_t|}}\text{.}
\end{align*}
If $\questimator$ does not return after this update, that means the measurement
operation returned $0$. Therefore, after the measurement, 
\begin{align*}
\Sigma_{t+1} &=
\frac{O^{xy}_0U^{t+1}_{xy}\Sigma_t}{\norm*{O^{xy}_0U^{t+1}_{xy}\Sigma_t}_2}\\
&= \frac{\sum_{i=2t+3}^{2m} \ket{i} + \sum_{\dedge{uv} \in S_{t+1}}
\ket{\dedge{uv}}}{\norm{\sum_{i=2t+3}^{2m} \ket{i} + \sum_{\dedge{uv} \in S_{t+1}}
\ket{\dedge{uv}}}_2}\\
&= \frac{\sum_{i=2t+3}^{2m} \ket{i} + \sum_{\dedge{uv} \in S_{t+1}}
\ket{\dedge{uv}}}{\sqrt{2m - 2(t+1) + \abs{S_{t+1}}}}
\end{align*}
and so the result holds for all $t \in \brac{m}$.
\end{proof}

\begin{lemma}
\label{lm:tlessex}
\[
\E{\bb} = \frac{\Tlk}{km}
\]
\end{lemma}
\begin{proof}
For $\bb$ to take a non-zero value, it must be returned from the measurement
step at some $t$ such that $f(t) = 1$, and the algorithm must not return before
time $t$. Condition on the value of $f$. For any $t$ such that $f(t) = 1$, let
$xy = \sigma_t$, so $S_t = \set{\dedge{xy},\dedge{yx}} \cup S_{t-1} \setminus
W_t$, where $W_t = \set{\dedge{wz} \in S_{t-1} : z \in \set{x,y}}$.

Suppose that the algorithm has not yet returned a value. By
Lemma~\ref{lm:stateinv},  \[
\Sigma_{t-1} = \frac{\sum_{i=2t-1}^{2m} \ket{i} + \sum_{\dedge{uv} \in S_{t-1}}
\ket{\dedge{uv}}}{\sqrt{2m - 2(t-1) + |S_{t-1}|}}
\]
and so
\begin{align*}
U_{xy}^t\Sigma_{t-1} &= \frac{\sum_{i=2t+1}^{2m} \ket{i} + \sum_{\dedge{uv}
\in S_t} \ket{\dedge{uv}}}{\sqrt{2m - 2(t-1) + \abs{S_{t-1}}}} +
\frac{\sum_{\dedge{uv} \in W_t} \ket{\dedge{uv}}}{\sqrt{2m - 2(t-1) + |S_{t-1}|}}\\
&= \frac{\sum_{i=2t+1}^{2m} \ket{i} + \sum_{\dedge{uv}
\in S_t} \ket{\dedge{uv}}}{\sqrt{2m - 2(t-1) + \abs{S_{t-1}}}} + \frac{\sum_{w \in W_t^+}
\paren*{\ket{\dedge{wx}} + \ket{\dedge{wy}}}}{\sqrt{2m - 2(t-1) + |S_{t-1}|}} \\
&\phantom{=} + \frac{\sum_{w \in W_t^-} \paren*{\ket{\dedge{wx}} + \ket{\dedge{wy}}}}{2\sqrt{2m - 2(t-1) + |S_{t-1}|}} +
\frac{\sum_{w \in W_t^-}\paren*{\ket{\dedge{wx}} -
\ket{\dedge{wy}}}}{2\sqrt{2m - 2(t-1) + |S_{t-1}|}}(-1)^{\Ibb(\dedge{wy} \in W_t)}
\end{align*}
where
\begin{align*}
W_t^+ &= \set{w \in V : \abs{\set{\dedge{wx}, \dedge{wy}} \cap W_t} = 2}\\
W_t^{-} &= \set{w \in V : \abs{\set{\dedge{wx}, \dedge{wy}} \cap W_t} = 1}\text{.}
\end{align*}

Therefore, conditioned on $f$ and the algorithm not having returned yet, after
the measurement (breaking the 1 case into two for clarity) \[
\bb = \begin{cases}
1 & \mbox{with probability $\frac{2\abs{W_t^+}}{2m - 2(t-1) + \abs{S_{t-1}}}$}\\
1 & \mbox{with probability $\frac{\abs{W_t^-}}{2(2m - 2(t-1) + \abs{S_{t-1}})}$}\\
-1 & \mbox{with probability $\frac{\abs{W_t^-}}{2(2m - 2(t-1) + \abs{S_{t-1}})}$}\\
0 & \mbox{with probability $\frac{2m - 2t + \abs{S_t}}{2m - 2(t-1) + \abs{S_{t-1}}}$.}
\end{cases}
\]
Now, let $t_i$ be the $i^\text{th}$ $t$ such that $f(t) = 1$, and let $l =
\abs{\set{t \in \brac{m} : f(t) = 1}}$. For each $i \in \brac{l}$, let $H^f_i$
be the expected value of $\bb$ when it is returned, conditioned on $f$ and on
the algorithm not returning at a time step before $t_i$. Then, by the above we
have \[
H^f_i = \frac{2\abs{W_{t_i}^+}}{2m - 2(t_i - 1) + \abs{S_{t_i-1}}} + \frac{2m -
2t_i + \abs{S_{t_i}}}{2m - 2(t_i - 1) + \abs{S_{t_i-1}}}H^f_{i+1}
\]
with $H^f_{l + 1}$ defined to be $0$. We will prove by reverse induction on $i$
that \[
H^f_i = \frac{2\sum_{j=i}^l \abs{W_{t_j}^+}}{2m - 2(t_{i} - 1) + \abs{S_{t_i-1}}} 
\]
for all $i \in \brac{l + 1}$. We have $H^f_{l + 1} = 0$ by definition. Now, for
any $i < l + 1$, suppose \[
H^f_{i+1} = \frac{2\sum_{j=i+1}^l \abs{W_{t_j}^+}}{2m - 2(t_{i+1}-1) +
\abs{S_{t_{i+1}-1}}}\text{.}
\]
Then
\begin{align*}
H^f_i &= \frac{2\abs{W_{t_i}^+}}{2m - 2(t_i - 1) + \abs{S_{t_i-1}}} + \frac{2m -
2{t_i} + \abs{S_{t_i}}}{2m - 2(t_i - 1) + \abs{S_{t_i -
1}}}\frac{2\sum_{j=i+1}^l \abs{W_{t_j}^+}}{2m - 2(t_{i+1} - 1) +
\abs{S_{t_{i+1} - 1}}}\\
&= \frac{2\sum_{j=i}^l \abs{W_{t_j}^+}}{2m - 2(t_{i} - 1) + \abs{S_{t_i - 1}}}
\end{align*}
as $\abs{S_{t_{i+1} - 1}} = \abs{S_{t_i}} + 2(t_{i+1} - t_{i} - 1)$, as $f(s) = 0$
when $t_i < s < t_{i+1}$, and so the set grows by two edges at each such $s$.

So we now have
\begin{align*}
\E{\bb|f} &= H_1^f\\
&= \frac{1}{m}\sum_{j=1}^l \abs{W_{t_j}^+}\\
&= \frac{1}{m}\sum_{t=1}^m \abs{W^+_t}
\end{align*}
by defining $W_t^+$ to be the empty set when $f(t) = 0$. And so \[
\E{\bb} = \frac{1}{m}\sum_{t=1}^m \E{\abs{W^+_t}}
\]

Now, recall that, when $f(t) = 1$, $W_{t}^+$ consists of the vertices $w$ in
$V$ such that, if $uv$ is the edge that arrives at time $t$, $wu$ and $wv$
arrive at times $s_1,s_2 < t$, and there is no time $s_1' \in (s_1,t) \cap
\brac{m}$ such that $f(s_1') = 1$ and $\sigma_{s'}$ is incident to $u$, nor
time $s_2' \in (s_2, t) \cap \brac{m}$ such that $f(s_2') = 1$ and
$\sigma_{s_2'}$ is incident to $v$.

If both $wu$ and $wv$ arrived before $uv$, the probability that this happens
conditioned on $f(t)=1$ is $(1 - 1/k)^{\degb{uvw} + \degb{uwv}} = \tlk_{wuv} +
\tlk_{wvu}$.

So \[
\E{\abs{W^+_t} \middle| f(t) = 1} = \sum_{w \in V}\paren*{\tlk_{wuv} +
\tlk_{wvu}}
\]
where $uv = \sigma_t$ and so \[
\E{\abs{W^+_t}} = \frac{1}{k} \sum_{w \in V}\paren*{\tlk_{wuv} +
\tlk_{wvu}}\text{.}
\]
As $\tlk_{xyz}$ is 0 whenever $yz$ is not an edge that appears in the stream,
this gives us
\begin{align*}
\E{\bb} &= \frac{1}{km}\sum_{\substack{t=1\\ uv \coloneqq \sigma_t}}^m \sum_{w
\in
V}\paren*{\tlk_{wuv} +
\tlk_{wvu}}\\
&= \frac{1}{km}\sum_{(u,v,w) \in V^3}\tlk_{uvw}\\
&= \frac{\Tlk}{km}
\end{align*}
completing the proof.
\end{proof}

\begin{lemma}
\label{lm:tlessvar}
\[
\Var{\bb} \le 1
\]
\end{lemma}
\begin{proof}
This follows immediately from the fact that $|\bb|$ is at most $1$.
\end{proof}

\begin{lemma}
\label{lm:questimate}
For any $\varepsilon, \delta \in (0,1\rbrack$, there is a quantum streaming
algorithm, using \[
O\paren*{\paren*{\frac{km}{T}}^2\log n \frac{1}{\varepsilon^2}\log
\frac{1}{\delta}}
\]
quantum and classical bits, that estimates $\Tlk$ to $\varepsilon T$
precision with probability $1-\delta$.
\end{lemma}
\begin{proof}
By taking the average of $\bT{\paren*{km/T\varepsilon}^2}$ copies of
$\questimator(k)$ and multiplying by $km$, we obtain an estimator with
expectation $\Tlk$ and variance at most $\varepsilon^2 T^2/4$. So by
Chebyshev's inequality the estimator will be within $\varepsilon T$ of $\Tlk$
with probability $3/4$. We can then repeat this $O(\log \frac{1}{\delta})$
times and take the median of our estimators to estimate $\Tlk$ to
$T\varepsilon$ precision with probability $1 - \delta$.

Each copy of $\questimator$ uses $2\ceil{\log n} + 1$ qubits and $\bO{\log n}$
classical bits, and so the result follows.
\end{proof}

\section{Classical Estimator}
For this algorithm we will use a hash function $g : V \rightarrow \bool$ such that \[
\E{g(v)} = 1/\sqrt{km}
\]
for each $v \in V$. We \emph{will} need to store this function, but instead of
making it fully independent we will make it pairwise independent, so this will
only require $O(\log n)$ bits.
\begin{algorithm}[H]
\caption{Classical estimator for $\Tgk$.}
\begin{algorithmic}[1]
\Procedure{ClassicalEstimator}{$k$}
\State $S \gets \emptyset$ 
\State $D \gets $ the empty map from $S$ to $\mathbb{N}$. 
\State $\Xb = 0$
\For{each update $uv$}
\For{$w \in V$}
\If{$\dedge{wu},\dedge{wv} \in S$}
\State $\Xb \gets \Xb + 1 - (1 - 1/k)^{D\brac*{\dedge{wu}} +
D\brac*{\dedge{wv}}}$ \Comment{Add
$\tgk_{wuv} + \tgk_{wvu}$ to $\Xb$.}
\EndIf
\EndFor
\For{$\dedge{xy} \in S$}
\If{$y \in \set{u,v}$}
\State $D\brac{\dedge{xy}} \gets D\brac{\dedge{xy}} + 1$
\EndIf
\EndFor
\If{$g(u) = 1$} 
\State Add $\dedge{uv}$ to $S$ with probability $\sqrt{k/m}$.
\State If $\dedge{uv}$ was added, set $D\brac{\dedge{uv}}=0$.
\EndIf
\If{$g(v) = 1$} 
\State Add $\dedge{vu}$ to $S$ with probability $\sqrt{k/m}$.
\State If $\dedge{vu}$ was added, set $D\brac{\dedge{vu}}=0$.
\EndIf
\EndFor
\State \return $\Xb$.
\EndProcedure
\end{algorithmic}
\end{algorithm}

\begin{lemma}
\label{lm:classtimatorspace}
The expected space usage of $\clastimator$ is $O(\log n)$ bits.
\end{lemma}
\begin{proof}
For each edge in $G$, $\clastimator$ will keep it with probability
$\frac{1}{\sqrt{km}} \times \sqrt{k/m} = \frac{1}{m}$ for each of its
endpoints. So the algorithm keeps $O(1)$ edges in expectation, along with a
counter (of size at most $m$) for each edge. The edge and the counter can each
be stored in $O(\log n)$ bits, and so the result follows.
\end{proof}

\begin{lemma}
\[
\E{\Xb} = \frac{\Tgk\sqrt{k}}{m^{3/2}}
\]
\end{lemma}
\begin{proof}
For any triangle $uvw$ such that $uv \bef uw \bef vw$, $\tgk$ will be added to
$\Xb$ iff $g(u) = 1$ (which happens with probability $1/\sqrt{km}$) and both
$uv$ and $uw$ are then kept by the algorithm (which happens with probability
$(\sqrt{k/m})^2$ conditioned on $g(u) = 1$).
\end{proof}

\begin{lemma}
\[
\Var{\Xb} \le 4\frac{\Tgk\sqrt{k}}{m^{3/2}}\Delta_E
\]
\end{lemma}
\begin{proof}
For each $(u,v,w) \in V^3$, let $\Xb_{uvw}$ be the contribution to $\Xb$ from
(possibly) adding $\tgk_{uvw}$. For each $u \in V$, let $\Xb_u = \sum_{(v,w)
\in V^2} X_{uvw}$. Then $\Xb = \sum_u \Xb_u$ and the $\Xb_u$ are all
independent, as each depends on $g(u)$ and the independent edge addition
events. So $\Var{\Xb} = \sum_{u \in V} \Var{\Xb_u}$.

Now, for each $u \in V$,
\begin{align*}
\Var{\Xb_u} &\le \E{\Xb_u^2}\\
&= \E{\paren*{\sum_{(v,w) \in V^2} \Xb_{uvw}}^2}\\
&= \E{\paren*{\sum_{\substack{(v,w) \in N(u)^2 :\\ vw \in E\\ uv \bef uw
\bef vw}} \Xb_{uvw}}^2}\\
&= \sum_{\substack{(v,w) \in N(u)^2 :\\ vw \in E\\ uv \bef uw \bef vw}}
\E{\Xb_{uvw}^2} +  \sum_{\substack{(v,w,x,y) \in N(u)^4 :\\ vw, xy \in E\\ uv
\bef uw \bef vw\\ ux \bef uy \bef xy\\ \abs{\set{v,w} \cap \set{x,y}} = 1}}
\E{\Xb_{uvw}\Xb_{xyy}} + \sum_{\substack{(v,w,x,y) \in N(u)^4 :\\ vw, xy \in
E\\ uv \bef uw \bef vw\\ ux \bef uy \bef xy\\ \set{v,w} \cap \set{x,y} =
\emptyset}} \E{\Xb_{uvw}\Xb_{xyy}}
\end{align*}

We will bound each of these three terms in turn. First, as the probability that
we add $\tgk_{uvw}$ to $\Xb$ is $\frac{1}{\sqrt{km}} \times
\paren*{\sqrt{k/m}}^2$, and $0 \le \tgk_{uvw} \le 1$, 
\begin{align*}
\sum_{\substack{(v,w) \in N(u)^2 :\\ vw \in E\\ uv \bef uw \bef vw}}
\E{\Xb_{uvw}^2} &\le \sum_{\substack{(v,w) \in N(u)^2 :\\ vw \in E\\ uv \bef uw
\bef vw}} \E{\Xb_{uvw}}\\
&= \frac{\Tgk_u\sqrt{k}}{m^{3/2}}\text{.}
\stepcounter{equation}\tag{\theequation}\label{eq:samet}
\end{align*}
Next, each triangle shares an edge with at most $\Delta_E$ other triangles, and
$\Xb_{uvw}\Xb_{uxy} > 0$ only if $g(u) = 1$ and all of $uv,uw,ux,uy$ are kept
by $\clastimator$, which occurs with probability $\frac{1}{\sqrt{km}} \times
\paren*{\sqrt{k/m}}^3 = k/m^2$ when there are exactly three distinct vertices
among $v,w,x,y$. So again using the fact that $0 \le \tgk_{uvw} \le 1$,
\begin{align*}
\sum_{\substack{(v,w,x,y) \in N(u)^4 :\\ vw, xy \in E\\ uv
\bef uw \bef vw\\ ux \bef uy \bef xy\\ \abs{\set{v,w} \cap \set{x,y}} = 1}}
\E{\Xb_{uvw}\Xb_{xyy}} &\le \sum_{\substack{(v,w) \in N(u)^2 :\\ vw \in E\\ uv
\bef uw \bef vw}} \Delta_E \tgk_{uvw} \frac{k}{m^2}\\ 
&= \frac{\Tgk_u k}{m^2}\Delta_E\text{.}
\stepcounter{equation}\tag{\theequation}\label{eq:onedge}
\end{align*}
For the final term, we will need the fact that $\tgk_{uvw} \le \frac{d_v +
d_w}{k}$. If $k \le d_v + d_w$, this follows immediately from the fact that
$\tgk_{uvw} \le 1$. Otherwise, if $\tgk_{uvw} \not= 0$,
\begin{align*}
\tgk_{uvw} &= 1 - (1 - 1/k)^{\degb{uv} + \degb{uw}}\\
&\le 1 - (1 - 1/k)^{d_v + d_w}\\
&= -\sum_{i = 1}^{d_v + d_w}\binom{d_v + d_w}{i}(-1/k)^i\\
&= \frac{d_v + d_w}{k} - \sum_{i = 2}^{d_v + d_w}\binom{d_v + d_w}{i}(-1/k)^i\\
&\le \frac{d_v + d_w}{k}
\end{align*}
as the terms of $ \sum_{i = 2}^{d_v + d_w}\binom{d_v + d_w}{i}(-1/k)^i$
alternate between positive and negative, and their magnitude is decreasing in
$i$ (as $k > d_v + d_w$), and they start positive, so the sum is non-negative.

Therefore, as for disjoint $\set{v,w},\set{x,y}$, if $\tgk_{uvw} \tgk_{uxy} >
0$,  the probability that $\Xb_{uvw}\Xb_{uxy} \not=0$ is $\frac{1}{\sqrt{km}}
\times (\sqrt{k/m})^4 = \frac{k^{3/2}}{m^{5/2}}$,
\begin{align*}
\sum_{\substack{(v,w,x,y) \in N(u)^4 :\\ vw, xy \in E\\ uv \bef uw \bef vw\\ ux
\bef uy \bef xy\\ \set{v,w} \cap \set{x,y} = \emptyset}} \E{\Xb_{uvw}\Xb_{xyy}}
&= \sum_{\substack{(v,w,x,y) \in N(u)^4 :\\ vw, xy \in E\\ uv \bef uw \bef vw\\
ux \bef uy \bef xy\\ \set{v,w} \cap \set{x,y} = \emptyset}}
\frac{k^{3/2}}{m^{5/2}} \tgk_{uvw} \tgk_{uxy}\\
&\le \sum_{\substack{(v,w) \in N(u)^2 :\\ vw \in E\\ uv \bef uw \bef vw}}
\frac{k^{3/2}}{m^{5/2}} \sum_{\substack{(x,y) \in N(u)^2 :\\ xy \in E\\ ux \bef
uy \bef xy}} \frac{d_x + d_y}{k}
\end{align*}
and, as for each $x \in N(u)$ there are at most $\Delta_E$ elements $y$ of $N(u)$ such that $uxy$ is a triangle, we have
\begin{align*}
\sum_{\substack{(x,y) \in N(u)^2 :\\ xy \in E\\ ux \bef uy \bef xy}} \frac{d_x
+ d_y}{k} &= \sum_{x \in N(u)}\sum_{\substack{y \in N(u) :\\ xy \in E\\ ux \bef
xy\\ uy \bef xy\\ d_y \le d_x}} \frac{d_x + d_y}{k}\\
&\le \sum_{x \in N(u)}\frac{2d_x}{k}\Delta_E\\
&\le \frac{2m}{k}\Delta_E
\end{align*}
so
\begin{align*}
\sum_{\substack{(v,w,x,y) \in N(u)^4 :\\ vw, xy \in E\\ uv \bef uw \bef vw\\ ux
\bef uy \bef xy\\ \set{v,w} \cap \set{x,y} = \emptyset}} \E{\Xb_{uvw}\Xb_{xyy}}
&\le \sum_{\substack{(v,w) \in N(u)^2 :\\ vw \in E\\ uv \bef uw \bef vw}}
\frac{2\sqrt{k}}{m^{3/2}}\Delta_E\\
&= 2 \frac{\Tgk_u\sqrt{k}}{m^{3/2}}\Delta_E\text{.}
\stepcounter{equation}\tag{\theequation}\label{eq:disj}
\end{align*}
Therefore, by adding \eqref{eq:samet}, \eqref{eq:onedge}, and \eqref{eq:disj},
\begin{align*}
\Var{\Xb_u} &\le \frac{\Tgk_u\sqrt{k}}{m^{3/2}} + \frac{\Tgk_u k}{m^2}\Delta_E +
2\frac{\Tgk_u\sqrt{k}}{m^{3/2}}\Delta_E\\
&\le 4\frac{\Tgk_u\sqrt{k}}{m^{3/2}}\Delta_E
\end{align*}
as $\Delta_E \ge 1$ and $k \le m$. The result then follows from summing over
all $u \in V$.
\end{proof}

\begin{lemma}
\label{lm:clastimate}
For any $\varepsilon, \delta \in (0,1\rbrack$, there is a classical streaming
algorithm, using \[
O\paren*{\frac{m^{3/2}}{T \sqrt{k}}\Delta_E\log
n\frac{1}{\varepsilon^2}\log \frac{1}{\delta} }
\]
bits of space in expectation, that estimates $\Tgk$ to $\varepsilon T$
precision with probability $1-\delta$.
\end{lemma}
\begin{proof}
As $\Tgk \le T$, we can take the average of
$\Theta\paren*{\frac{1}{\varepsilon^2}\frac{m^{3/2}}{T
\sqrt{k}}\Delta_E}$ copies of $\clastimator(k)$ and multiply by
$\frac{m^{3/2}}{\sqrt{k}}$ to obtain an estimator with expectation $\Tgk$ and
variance at most $\varepsilon^2 (\Tgk)^2/4 \le \varepsilon^2 T^2/4$. So by
Chebyshev's inequality the estimator will be within $\varepsilon T$ of $\Tgk$
with probability $3/4$. We can then repeat this $O(\log \frac{1}{\delta})$
times and take the median of our estimators to estimate $\Tgk$ to
$T\varepsilon$ precision with probability $1 - \delta$.

By Lemma~\ref{lm:classtimatorspace}, each copy of $\clastimator$ will require
$\bO{\log n}$ bits of space in expectation, and so the result follows.
\end{proof}

\section{Hybrid Quantum-Classical Algorithm}
By combining our quantum and classical estimators, we may now prove
Theorem~\ref{thm:main}.
\main*
\begin{proof}
Let \[
k = \frac{T^{2/5}}{m^{1/5}}\Delta_E^{2/5}\text{.}
\]
By Lemmas~\ref{lm:questimate} and~\ref{lm:clastimate}, there are algorithms for
estimating each of $\Tlk$, $\Tgk$ to precision $\varepsilon T/2$ with
probability $1-\delta/2$, using \[
O\paren*{ \frac{m^{8/5}}{T^{6/5}}\Delta_E^{4/5}\log n
\frac{1}{\varepsilon^2}\log \frac{1}{\delta}}
\]
quantum and classical bits in expectation. If we then sum these estimators they
will be within $\varepsilon T$ of $T$ with probability $1 - \delta$, by taking
a union bound.
\end{proof}

\section*{Acknowledgements}
The author would like to thank Scott Aaronson for suggesting the technique of
using ``dummy variables'' to construct a superposition in the stream.

The author was supported by the National Science Foundation (NSF) under Grant
Number CCF-1751040 (CAREER). Also supported by Laboratory Directed Research and
Development program at Sandia National Laboratories, a multimission laboratory
managed and operated by National Technology and Engineering Solutions of
Sandia, LLC., a wholly owned subsidiary of Honeywell International, Inc., for
the U.S. Department of Energy’s National Nuclear Security Administration under
contract DE-NA-0003525. Also supported by the U.S. Department of Energy, Office
of Science, Office of Advanced Scientific Computing Research, Accelerated
Research in Quantum Computing program.
\newpage
\bibliographystyle{alpha}
\bibliography{refs}

\end{document}